\documentclass[english]{article}
\usepackage[T1]{fontenc}
\usepackage[latin9]{inputenc}
\usepackage{geometry}
\usepackage{xcolor}
\usepackage{pdfcolmk}
\usepackage{amsthm}
\usepackage{amsmath}
\usepackage{amssymb}
\usepackage{esint}
\usepackage[numbers]{natbib}
\PassOptionsToPackage{normalem}{ulem}
\usepackage{ulem}

\makeatletter

\providecolor{lyxadded}{rgb}{0,0,1}
\providecolor{lyxdeleted}{rgb}{1,0,0}
\theoremstyle{plain}
\newtheorem{thm}{\protect\theoremname}

\usepackage{nips15submit_e,times}
\usepackage{hyperref}
\usepackage{url}

\nipsfinalcopy

\makeatother

\usepackage{babel}
\providecommand{\theoremname}{Theorem}

\begin{document}
\global\long\def\bb{\bm{\beta}}
\global\long\def\bt{\bm{\theta}}
\global\long\def\bmu{\bm{\mu}}
\global\long\def\bl{\bm{\lambda}}
\global\long\def\be{\bm{\eta}}
\global\long\def\E{\mathbb{E}}
\global\long\def\by{\bm{y}}
\global\long\def\bys{\bm{y^{\star}}}
\global\long\def\bS{\bm{\Sigma}}
\global\long\def\N{\mathcal{N}}
\global\long\def\I{\mathds{1}}
\global\long\def\A{\mathcal{A}}
\global\long\def\Ml{\mathcal{M}_{l}}
\global\long\def\Q{\mathcal{Q}}
\global\long\def\epinf{EP^{\infty}}
\global\long\def\O{\mathcal{O}}

\title{Bounding errors of Expectation-Propagation}

\author{Guillaume Dehaene\\
University of Geneva\\
\texttt{guillaume.dehaene@gmail.com}\And Simon Barthelmé\\
CNRS, Gipsa-lab\\
\texttt{simon.barthelme@gipsa-lab.fr}}
\maketitle
\begin{abstract}
Expectation Propagation is a very popular algorithm for variational
inference, but comes with few theoretical guarantees. In this article,
we prove that the approximation errors made by EP can be bounded.
Our bounds have an asymptotic interpretation in the number $n$ of
datapoints, which allows us to study EP's convergence with respect
to the true posterior. In particular, we show that EP converges at
a rate of $\O(n^{-2})$ for the mean, up to an order of magnitude
faster than the traditional Gaussian approximation at the mode. We
also give similar asymptotic expansions for moments of order 2 to
4, as well as excess Kullback-Leibler cost (defined as the additional
KL cost incurred by using EP rather than the ideal Gaussian approximation).
All these expansions highlight the superior convergence properties
of EP. Our approach for deriving those results is likely applicable
to many similar approximate inference methods. In addition, we introduce
bounds on the moments of log-concave distributions that may be of
independent interest. 
\end{abstract}

\section*{Introduction}

Expectation Propagation (EP, \citealp{Minka:EP}) is an efficient
approximate inference algorithm that is known to give good approximations,
to the point of being almost exact in certain applications \citep{KussRasmussen:AssessingApproxInfGP,NickishRasmussen:ApproxGaussianProcClass}.
It is surprising that, while the method is empirically very successful,
there are few theoretical guarantees on its behavior. Indeed, most
work on EP has focused on efficiently implementing the method in various
settings. Theoretical work on EP mostly represents new justifications
of the method which, while they offer intuitive insight, do not give
mathematical proofs that the method behaves as expected. One recent
breakthrough is due to \citet{DehaeneBarthelme:EPLargeDataLimit}
who prove that, in the large data-limit, the EP iteration behaves
like a Newton search and its approximation is asymptotically exact.
However, it remains unclear how good we can expect the approximation
to be when we have only finite data. In this article, we offer a characterization
of the quality of the EP approximation in terms of the worst-case
distance between the true and approximate mean and variance.

When approximating a probability distribution $p(x)$ that is, for
some reason, close to being Gaussian, a natural approximation to use
is the Gaussian with mean equal to the mode (or argmax) of $p(x)$
and with variance the inverse log-Hessian at the mode. We call it
the Canonical Gaussian Approximation (CGA), and its use is usually
justified by appealing to the Bernstein-von Mises theorem, which shows
that, in the limit of a large amount of independent observations,
posterior distributions tend towards their CGA. This powerful justification,
and the ease with which the CGA is computed (finding the mode can
be done using Newton methods) makes it a good reference point for
any method like EP which aims to offer a better Gaussian approximation
at a higher computational cost. In section \ref{sec:Background, EP algorithm, etc},
we introduce the CGA and the EP approximation. In section \ref{sec: theory and theoretical bounds},
we give our theoretical results bounding the quality of EP approximations.

\section{Background\label{sec:Background, EP algorithm, etc}}

In this section, we present the CGA and give a short introduction
to the EP algorithm. In-depth descriptions of EP can be found in \citet{Minka:DivMeasuresMP,Seeger:EPExpFam,BishopPRML,Raymond:ExpectationPropagation}.

\subsection{The Canonical Gaussian Approximation }

What we call here the CGA is perhaps the most common approximate inference
method in the machine learning cookbook. It is often called the ``Laplace
approximation'', but this is a misnomer: the Laplace approximation
refers to approximating the integral $\int p$ from the integral of
the CGA. The reason the CGA is so often used is its compelling simplicity:
given a target distribution $p(x)=\exp\left(-\phi\left(x\right)\right)$,
we find the mode $x^{\star}$ and compute the second derivatives of
$\phi$ at $x^{\star}$:
\begin{eqnarray*}
x^{\star} & = & \mbox{argmin}\phi(x)\\
\beta^{\star} & = & \phi''\left(x^{\star}\right)
\end{eqnarray*}

to form a Gaussian approximation $q(x)=\mathcal{N}\left(x|x^{\star},\frac{1}{\beta^{\star}}\right)\approx p(x)$.
The CGA is effectively just a second-order Taylor expansion, and its
use is justified by the Bernstein-von Mises theorem \citep{DasGupta:AsymptoticTheoryStats},
which essentially says that the CGA becomes exact in the large-data
(large-$n$) asymptotic limit. Roughly, if $p_{n}(x)\propto\prod_{i=1}^{n}p\left(y_{i}|x\right)p_{0}\left(x\right)$,
where $y_{1}\ldots y_{n}$ represent independent datapoints, then
$\lim_{n\rightarrow\infty}p_{n}\left(x\right)=\mathcal{N}\left(x|x_{n}^{\star},\frac{1}{\beta_{n}^{\star}}\right)$
in total variation.

\subsection{CGA vs Gaussian EP}

Gaussian EP, as its name indicates, provides an alternative way of
computing a Gaussian approximation to a target distribution. There
is broad overlap between the problems where EP can be applied and
the problems where the CGA can be used, with EP coming at a higher
cost. Our contribution is to show formally that the higher computational
cost for EP may well be worth bearing, as EP approximations can outperform
CGAs by an order of magnitude. To be specific, we focus on the \emph{moment
estimates }(mean and covariance) computed by EP and CGA, and derive
bounds on their distance to the true mean and variance of the target
distribution. Our bounds have an asymptotic interpretation, and under
that interpretation we show for example that the mean returned by
EP is within an order of $\O\left(n^{-2}\right)$ of the true mean,
where $n$ is the number of datapoints. For the CGA, which uses the
mode as an estimate of the mean, we exhibit a $\O\left(n^{-1}\right)$
upper bound, and we compute the error term responsible for this $\O\left(n^{-1}\right)$
behavior. This enables us to show that, in the situations in which
this error is indeed $\O\left(n^{-1}\right)$, EP is better than the
CGA.

\subsection{The EP algorithm\label{sub:The-EP-algorithm}}

We consider the task of approximating a probability distribution over
a random-variable $\mathcal{X}:$ $p(x)$, which we call the \emph{target
distribution}. $\mathcal{X}$ can be high-dimensional, but for simplicity,
we focus on the one-dimensional case. One important hypothesis that
makes EP feasible is that $p(x)$ factorizes into $n$ \emph{simple
}factor terms: 
\[
p(x)=\prod_{i}f_{i}(x)
\]

EP proposes to approximate each $f_{i}(x)$ (usually referred to as
\emph{sites}) by a Gaussian function $q_{i}(x)$ (referred to as the
\emph{site-approximations}). It is convenient to use the parametrization
of Gaussians in terms of natural parameters:
\[
q_{i}\left(x|r_{i},\beta_{i}\right)\propto\exp\left(r_{i}x-\beta_{i}\frac{x^{2}}{2}\right)
\]

which makes some of the further computations easier to understand.
Note that EP could also be used with other exponential approximating
families. These Gaussian approximations are computed iteratively.
Starting from a current approximation $\left(q_{i}^{t}\left(x|r_{i}^{t},\beta_{i}^{t}\right)\right)$,
we select a site for update with index i. We then:
\begin{itemize}
\item Compute the \emph{cavity} distribution $q_{-i}^{t}(x)\propto\prod_{j\neq1}q_{j}^{t}(x)$.
This is very easy in natural parameters:
\[
q_{-i}(x)\propto\exp\left(\left(\sum_{j\neq i}r_{j}^{t}\right)x-\left(\sum_{j\neq i}\beta_{j}^{t}\right)\frac{x^{2}}{2}\right)
\]

\item Compute the \emph{hybrid} distribution $h_{i}^{t}(x)\propto q_{-i}^{t}(x)f_{i}(x)$
and its mean and variance
\item Compute the Gaussian which minimizes the Kullback-Leibler divergence
to the hybrid, ie the Gaussian with same mean and variance:
\[
\mathcal{P}(h_{i}^{t})=\underset{q}{\mbox{argmin}}\left(KL\left(h_{i}^{t}|q\right)\right)
\]

\item Finally, update the approximation of $f_{i}$:
\[
q_{i}^{t+1}=\frac{\mathcal{P}(h_{i}^{t})}{q_{-i}^{t}}
\]
where the division is simply computed as a subtraction between natural
parameters
\end{itemize}
We iterate these operations until a fixed point is reached, at which
point we return a Gaussian approximation of $p(x)\approx\prod q_{i}(x)$.

\subsection{The ``EP-approximation''}

In this work, we will characterize the quality of an EP approximation
of $p(x)$. We define this to be any fixed point of the iteration
presented in section \ref{sub:The-EP-algorithm}, which could all
be returned by the algorithm. It is known that EP will have at least
one fixed-point \citep{Minka:EP}, but it is unknown under which conditions
the fixed-point is unique. We conjecture that, when all sites are
log-concave (one of our hypotheses to control the behavior of EP),
it is in fact unique but we can't offer a proof yet. If $p\left(x\right)$
isn't log-concave, it is straightforward to construct examples in
which EP has multiple fixed-points. These open questions won't matter
for our result because we will show that all fixed-points of EP (should
there be more than one) produce a good approximation of $p\left(x\right)$.

Fixed points of EP have a very interesting characterization. If we
note $q_{i}^{*}$ the site-approximations at a given fixed-point,
$h_{i}^{*}$ the corresponding hybrid distributions, and $q^{*}$
the global approximation of $p(x)$, then the mean and variance of
all the hybrids and $q^{*}$ is the same%
\footnote{For non-Gaussian approximations, the expected values of all sufficient
statistics of the exponential family are equal.%
}. As we will show in section \ref{sub:Computing-bounds-on EP }, this
leads to a very tight bound on the possible positions of these fixed-points.

\subsection{Notation}

We will use repeatedly the following notation. $p(x)=\prod_{i}f_{i}(x)$
is the target distribution we want to approximate. The sites $f_{i}(x)$
are each approximated by a Gaussian site-approximation $q_{i}(x)$
yielding an approximation to $p(x)\approx q(x)=\prod_{i}q_{i}(x)$.
The hybrids $h_{i}(x)$ interpolate between $q(x)$ and $p(x)$ by
replacing one site approximation $q_{i}(x)$ with the true site $f_{i}(x)$.

Our results make heavy use of the log-functions of the sites and the
target distribution. We note $\phi_{i}(x)=-\log\left(f_{i}(x)\right)$
and $\phi_{p}(x)=-\log\left(p(x)\right)=\sum\phi_{i}(x)$. We will
introduce in section \ref{sec: theory and theoretical bounds} hypotheses
on these functions. Parameter $\beta_{m}$ controls their minimum
curvature and parameters $K_{d}$ control the maximum $d^{th}$ derivative.

We will always consider fixed-points of EP, where the mean and variance
under all hybrids and $q(x)$ is identical. We will note these common
values: $\mu_{EP}$ and $v_{EP}$. We will also refer to the third
and fourth centered moment of the hybrids, denoted by $m_{3}^{i},m_{4}^{i}$
and to the fourth moment of $q(x)$ which is simply $3v_{EP}^{2}$.
We will show how all these moments are related to the true moments
of the target distribution which we will note $\mu,v$ for the mean
and variance, and $m_{3}^{p},m_{4}^{p}$ for the third and fourth
moment. We also investigate the quality of the CGA: $\mu\approx x^{\star}$
and $v\approx\left[\phi_{p}^{''}(x^{\star})\right]^{-1}$ where $x^{\star}$
is the the mode of $p(x)$.

\section{Results \label{sec: theory and theoretical bounds}}

In this section, we will give tight bounds on the quality of the EP
approximation (ie: of fixed-points of the EP iteration). Our results
lean on the properties of log-concave distributions \citep{SaumardWellner:LogConcavityReview}.
In section \ref{sub:Log-concave-distributions-are constrained}, we
introduce new bounds on the moments of log-concave distributions.
The bounds show that those distributions are in a certain sense close
to being Gaussian. We then apply these results to study fixed points
of EP, where they enable us to compute bounds on the distance between
the mean and variance of the true distribution $p(x)$ and of the
approximation given by EP, which we do in section \ref{sub:Computing-bounds-on EP }. 

Our bounds require us to assume that all sites $f_{i}(x)$ are $\beta_{m}$-strongly
log-concave with slowly-changing log-function. That is, if we note
$\phi_{i}(x)=-\log\left(f_{i}(x)\right)$:
\begin{eqnarray}
\forall i\ \forall x\ \phi_{i}^{''}(x) & \geq & \beta_{m}>0\label{eq:hypothesis-LC-strict}\\
\forall i\ \forall d\in[3,4,5,6]\ \left|\phi_{i}^{(d)}(x)\right| & \leq & K_{d}\label{eq:hypothesis-slow}
\end{eqnarray}

The target distribution $p(x)$ then inherits those properties from
the sites. Noting $\phi_{p}(x)=-\log\left(p(x)\right)=\sum_{i}\phi_{i}(x)$,
then $\phi_{p}$ is $n\beta_{m}$-strongly log-concave and its higher
derivatives are bounded:
\begin{eqnarray}
\forall x,\ \phi_{p}^{''}(x) & \geq & n\beta_{m}\label{eq: strictly LC for target distribution}\\
\forall d\in[3,4,5,6]\ \left|\phi_{p}^{(d)}(x)\right| & \leq & nK_{d}\label{eq: slowly changing log for target distributoin}
\end{eqnarray}

A natural concern here is whether or not our conditions on the sites
are of practical interest. Indeed, strongly-log-concave likelihoods
are rare. We picked these strong regularity conditions because they
make the proofs relatively tractable (although still technical and
long). The proof technique carries over to more complicated, but more
realistic, cases. One such interesting generalization consists of
the case in which $p(x)$ and all hybrids at the fixed-point are log-concave
with slowly changing log-functions (with possibly differing constants).
In such a case, while the math becomes more unwieldy, similar bounds
as ours can be found, greatly extending the scope of our results.
The results we present here should thus be understood as a stepping
stone and not as the final word on the quality of the EP approximation:
we have focused on providing a rigorous but extensible proof.

\subsection{Log-concave distributions are strongly constrained\label{sub:Log-concave-distributions-are constrained}}

Log-concave distributions have many interesting properties. They are
of course unimodal, and the family is closed under both marginalization
and multiplication. For our purposes however, the most important property
is a result due to Brascamp and Lieb \citep{BrascampLieb:BestConstantsYoungsIneq},
which bounds their even moments. We give here an extension in the
case of log-concave distributions with slowly changing log-functions
(as quantified by eq. \eqref{eq:hypothesis-slow}). Our results show
that these are close to being Gaussian.

The Brascamp-Lieb inequality states that, if $LC(x)\propto\exp\left(-\phi(x)\right)$
is $\beta_{m}$-strongly log-concave (ie: $\phi^{''}(x)\geq\beta_{m}$),
then centered even moments of $LC$ are bounded by the corresponding
moments of a Gaussian with variance $\beta_{m}^{-1}$. If we note
these moments $m_{2k}$ and $\mu_{LC}=E_{LC}(x)$ the mean of $LC$:
\begin{eqnarray}
m_{2k} & = & E_{LC}\left(\left(x-\mu_{LC}\right)^{2k}\right)\nonumber \\
m_{2k} & \leq & (2k-1)!!\beta_{m}^{-k}\label{eq:Brascamp-Lieb inequality-1-1}
\end{eqnarray}

where $(2k-1)!!$ is the double factorial: the product of all odd
terms from 1 to $2k-1$. $3!!=3$, $5!!=15$, $7!!=105$, etc. This
result can be understood as stating that a log-concave distribution
must have a small variance, but doesn't generally need to be close
to a Gaussian.

With our hypothesis of slowly changing log-functions, we were able
to improve on this result. Our improved results include a bound on
\emph{odd }moments, as well as first order expansions of even moments
(eqs. \eqref{eq:bound on phi-prime}-\eqref{eq:first order of m4}). 

Our extension to the Brascamp-Lieb inequality is as follows. If $\phi$
is slowly changing in the sense that some of its higher derivatives
are bounded, as per eq. \ref{eq:hypothesis-slow}, then we can give
a bound on $\phi^{'}(\mu_{LC})$ (showing that $\mu_{LC}$ is close
to the mode $x^{\star}$ of $LC$, see eqs. \eqref{eq: restating another eq for p(x)}
to \eqref{eq: bounding mode-mean distance}) and $m_{3}$ (showing
that $LC$ is mostly symmetric):
\begin{eqnarray}
\left|\phi^{'}(\mu_{LC})\right| & \leq & \frac{K_{3}}{2\beta_{m}}\label{eq:bound on phi-prime}\\
\left|m_{3}\right| & \leq & \frac{2K_{3}}{\beta_{m}^{3}}\label{eq:bound on M3}
\end{eqnarray}

and we can compute the first order expansions of $m_{2}$ and $m_{4}$,
and bound the errors in terms of $\beta_{m}$ and the $K$'s :
\begin{eqnarray}
\left|m_{2}^{-1}-\phi^{''}(\mu_{LC})\right| & \leq & \frac{K_{3}^{2}}{\beta_{m}^{2}}+\frac{K_{4}}{2\beta_{m}}\label{eq:first order of m2}\\
\left|\phi^{''}(\mu_{LC})m_{4}-3m_{2}\right| & \leq & \frac{19}{2}\frac{K_{3}^{2}}{\beta_{m}^{4}}+\frac{5}{2}\frac{K_{4}}{\beta_{m}^{3}}\label{eq:first order of m4}
\end{eqnarray}

With eq. \eqref{eq:first order of m2} and \eqref{eq:first order of m4},
we see that $m_{2}\approx\left(\phi^{''}(\mu_{LC})\right)^{-1}$ and
$m_{4}\approx3\left(\phi^{''}(\mu_{LC})\right)^{-2}$ and, in that
sense, that $LC(x)$ is close to the Gaussian with mean $\mu_{LC}$
and inverse-variance $\phi^{''}\left(\mu_{LC}\right)$.

These expansions could be extended to further orders and similar formulas
can be found for the other moments of $LC(x)$: for example, any odd
moments can be bounded by $\left|m_{2k+1}\right|\leq C_{k}K_{3}\beta_{m}^{-(k+1)}$
(with $C_{k}$ some constant) and any even moment can be found to
have first-order expansion: $m_{2k}\approx\left(2k-1\right)!!\left(\phi^{''}(\mu_{LC})\right)^{-k}$.
The proof, as well as more detailed results, can be found in the Supplement.

Note how our result relates to the Bernstein-von Mises theorem, which
says that, in the limit of a large amount of observations, a posterior
$p(x)$ tends towards its CGA. If we consider the posterior obtained
from $n$ likelihood functions that are all log-concave and slowly
changing, our results show the slightly different result that the
moments of that posterior are \emph{close }to those of a Gaussian
with mean $\mu_{LC}$ (instead of $x_{LC}^{\star}$) and inverse-variance
$\phi^{''}\left(\mu_{LC}\right)$ (instead of $\phi^{''}\left(x_{LC}^{\star}\right)$)
. This point is critical. While the CGA still ends up capturing the
limit behavior of $p$, as $\mu_{LC}\rightarrow x^{\star}$ in the
large-data limit (see eq. \eqref{eq: bounding mode-mean distance}
below), an approximation that would return the Gaussian approximation
at $\mu_{LC}$ would be better. This is essentially what EP does,
and this is how it improves on the CGA.

\subsection{Computing bounds on EP approximations\label{sub:Computing-bounds-on EP }}

In this section, we consider a given EP fixed-point $q_{k}^{*}\left(x|r_{i},\beta_{i}\right)$
and the corresponding approximation of $p(x)$: $q^{*}\left(x|r=\sum r_{i},\beta=\sum\beta_{i}\right)$.
We will show that the expected value and variance of $q^{*}$(resp.
$\mu_{EP}$ and $v_{EP}$) are close to the true mean and variance
of $p$ (resp. $\mu$ and $v$), and also investigate the quality
of the CGA ($\mu\approx x^{\star}$, $v\approx\left[\phi_{p}^{''}(x^{\star})\right]^{-1}$). 

Under our assumptions on the sites (eq. \eqref{eq:hypothesis-LC-strict}
and \eqref{eq:hypothesis-slow}), we are able to derive bounds on
the quality of the EP approximation. The proof is quite involved and
long, and we will only present it in the Supplement. In the main text,
we give a partial version: we detail the first step of the demonstration,
which consists of computing a rough bound on the distance between
the true mean $\mu$, the EP approximation $\mu_{EP}$ and the mode
$x^{\star}$, and give an outline of the rest of the proof.

Let's show that $\mu$, $\mu_{EP}$ and $x^{\star}$ are all close
to one another. We start from eq. \eqref{eq:bound on phi-prime} applied
to $p(x)$:
\begin{equation}
\left|\phi_{p}^{'}(\mu)\right|\leq\frac{K_{3}}{2\beta_{m}}\label{eq: restating another eq for p(x)}
\end{equation}
which tells us that $\phi_{p}^{'}(\mu)\approx0$. $\mu$ must thus
be close to $x^{\star}$. Indeed:
\begin{eqnarray}
\left|\phi_{p}^{'}(\mu)\right| & = & \left|\phi_{p}^{'}(\mu)-\phi_{p}^{'}(x^{\star})\right|\label{eq: from phi-prime bound to e-m START}\\
 & = & \left|\phi_{p}^{''}\left(\xi\right)\left(\mu-x^{\star}\right)\right|\ \xi\in[\mu,x^{\star}]\nonumber \\
 & \geq & \left|\phi_{p}^{''}\left(\xi\right)\right|\left|\mu-x^{\star}\right|\nonumber \\
 & \geq & n\beta_{m}\left|\mu-x^{\star}\right|\label{eq:from phi-prime bound to e-m END}
\end{eqnarray}

Combining eq. \eqref{eq: restating another eq for p(x)} and \eqref{eq:from phi-prime bound to e-m END},
we finally have:
\begin{equation}
\left|\mu-x^{\star}\right|\leq n^{-1}\frac{K_{3}}{2\beta_{m}^{2}}\label{eq: bounding mode-mean distance}
\end{equation}
Let's now show that $\mu_{EP}$ is also close to $x^{\star}$. We
proceed similarly, starting from eq. \eqref{eq:bound on phi-prime}
but applied to all hybrids $h_{i}(x)$:
\begin{equation}
\forall i\ \left|\phi_{i}^{'}(\mu_{EP})+\beta_{-i}\mu_{EP}-r_{-i}\right|\leq n^{-1}\frac{K_{3}}{2\beta_{m}}\label{eq: restating an equation for hybrids}
\end{equation}
which is not really equivalent to eq. \eqref{eq: restating another eq for p(x)}
yet. Recall that $q(x|r,\beta)$ has mean $\mu_{EP}$: we thus have:
$r=\beta\mu_{EP}$. Which gives:
\begin{eqnarray}
\left(\sum_{i}\beta_{-i}\right)\mu_{EP} & = & \left((n-1)\beta\right)\mu_{EP}\nonumber \\
 & = & (n-1)r\nonumber \\
 & = & \sum_{i}r_{-i}\label{eq: sum of beta_-i vs sum of r_-i}
\end{eqnarray}
If we sum all terms in eq. \eqref{eq: restating an equation for hybrids},
the $\beta_{-i}\mu_{EP}$ and $r_{-i}$ thus cancel, leaving us with:
\begin{equation}
\left|\phi_{p}^{'}(\mu_{EP})\right|\leq\frac{K_{3}}{2\beta_{m}}\label{eq: equivalent to earlier equation}
\end{equation}
which is equivalent to eq. \eqref{eq: restating another eq for p(x)}
but for $\mu_{EP}$ instead of $\mu$. This shows that $\mu_{EP}$
is, like $\mu$, close to $x^{\star}$:
\begin{equation}
\left|\mu_{EP}-x^{\star}\right|\leq n^{-1}\frac{K_{3}}{2\beta_{m}^{2}}\label{eq: bound of e_ep - m}
\end{equation}
At this point, we can show that, since they are both close to $x^{\star}$
(eq. \eqref{eq: bounding mode-mean distance} and \eqref{eq: bound of e_ep - m}),
$\mu=\mu_{EP}+\O\left(n^{-1}\right)$, which constitutes the first
step of our computation of bounds on the quality of EP.

After computing this, the next step is evaluating the quality of the
approximation of the variance, via computing $\left|v^{-1}-v_{EP}^{-1}\right|$
for EP and $\left|v^{-1}-\phi_{p}^{''}(x^{\star})\right|$ for the
CGA, from eq. \eqref{eq:first order of m2}. In both cases, we find:
\begin{eqnarray}
v^{-1} & = & v_{EP}^{-1}+\O\left(1\right)\label{eq: quality of EP approx of precision}\\
 & = & \phi_{p}^{''}(x^{\star})+\O\left(1\right)\label{eq:quality of CGA approx of precision}
\end{eqnarray}
Since $v^{-1}$ is of order $n$, because of eq. \eqref{eq:Brascamp-Lieb inequality-1-1}
(Brascamp-Lieb upper bound on variance), this is a decent approximation:
the relative error is of order $n^{-1}$.

We can find similarly that both EP and CGA do a good job of finding
a good approximation of the fourth moment of $p$: $m_{4}$. For EP
this means that the fourth moment of each hybrid and of $q$ are a
close match:
\begin{eqnarray}
\forall i\ m_{4} & \approx & m_{4}^{i}\approx3v_{EP}^{2}\label{eq: fourth moment is matched by EP}\\
 & \approx & 3\left(\phi_{p}^{''}(m)\right)^{-2}\label{eq: fourth moment is matched by CGA}
\end{eqnarray}
In contrast, the third moment of the hybrids doesn't match at all
the third moment of $p$, but their sum does ! 
\begin{equation}
m_{3}\approx\sum_{i}m_{3}^{i}\label{eq: sum of third hybrid moment matches}
\end{equation}
Finally, we come back to the approximation of $\mu$ by $\mu_{EP}$.
These obey two very similar relationships:

\begin{eqnarray}
\phi_{p}^{'}(\mu)+\phi_{p}^{(3)}(\mu)\frac{v}{2} & = & \O\left(n^{^{-1}}\right)\label{eq: first order expansion for mean of P}\\
\phi_{p}^{'}(\mu_{EP})+\phi_{p}^{(3)}(\mu_{EP})\frac{v_{EP}}{2} & = & \O\left(n^{^{-1}}\right)
\end{eqnarray}

Since $v=v_{EP}+\O\left(n^{-2}\right)$ (a slight rephrasing of eq.
\eqref{eq: quality of EP approx of precision}), we finally have:
\begin{equation}
\mu=\mu_{EP}+\O\left(n^{-2}\right)
\end{equation}

We summarize the results in the following theorem:
\begin{thm}
Characterizing fixed-points of EP

Under the assumptions given by eq. \eqref{eq:hypothesis-LC-strict}
and \eqref{eq:hypothesis-slow} (log-concave sites with slowly changing
log), we can bound the quality of the EP approximation and the CGA:
\begin{eqnarray*}
\left|\mu-x^{*}\right| & \leq & n^{-1}\frac{K_{3}}{2\beta_{m}^{2}}\\
\left|\mu-\mu_{EP}\right| & \leq & B_{1}(n)=\O\left(n^{-2}\right)\\
\left|v^{-1}-\phi_{p}^{''}(x^{*})\right| & \leq & \frac{2K_{3}^{2}}{\beta_{m}^{2}}+\frac{K_{4}}{2\beta_{m}}\\
\left|v^{-1}-v_{EP}^{-1}\right| & \leq & B_{2}(n)=\O\left(1\right)
\end{eqnarray*}

We give the full expression for the bounds $B_{1}$ and $B_{2}$ in
the Supplement
\end{thm}
Note that the order of magnitude of the bound on $\left|\mu-x^{\star}\right|$
is the best possible, because it is attained for certain distributions.
For example, consider a Gamma distribution with natural parameters
$\left(n\alpha,n\beta\right)$ whose mean $\frac{\alpha}{\beta}$
is approximated at order $n^{-1}$ by its mode $\frac{\alpha}{\beta}-\frac{1}{n\beta}$.
More generally, from eq. \eqref{eq: first order expansion for mean of P},
we can compute the first order of the error: 
\begin{equation}
\mu-m\approx-\frac{\phi_{p}^{(3)}(\mu)}{\phi_{p}^{''}(\mu)}\frac{v}{2}\approx-\frac{1}{2}\frac{\phi_{p}^{(3)}(\mu)}{\left[\phi_{p}^{''}(\mu)\right]^{2}}
\end{equation}
which is the term causing the order $n^{-1}$ error. Whenever this
term is significant, it is thus safe to conclude that EP improves
on the CGA.

Also note that, since $v^{-1}$ is of order $n$, the relative error
for the $v^{-1}$ approximation is of order $n^{-1}$ for both methods.
Despite having a convergence rate of the same order, the EP approximation
is demonstrably better than the CGA, as we show next. Let us first
see why the approximation for $v^{-1}$ is only of order 1 for both
methods. The following relationship holds:
\begin{equation}
v^{-1}=\phi_{p}^{''}(\mu)+\phi_{p}^{(3)}(\mu)\frac{m_{3}^{p}}{2v}+\phi_{p}^{(4)}(\mu)\frac{m_{4}^{p}}{3!v}+\O\left(n^{-1}\right)\label{eq: the detailled v^-1 relationship for p}
\end{equation}
In this relationship, $\phi_{p}^{''}(\mu)$ is an order $n$ term
while the rest are order 1. If we now compare this to the CGA approximation
of $v^{-1}$, we find that it fails at multiple levels. First, it
completely ignores the two order 1 terms, and then, because it takes
the value of $\phi_{p}^{''}$ at $x^{\star}$ which is at a distance
of $\O\left(n^{-1}\right)$ from $\mu$, it adds another order 1 error
term (since $\phi_{p}^{(3)}=\O\left(n\right)$). The CGA is thus adding
quite a bit of error, even if each component is of order 1.

Meanwhile, $v_{EP}$ obeys a relationship similar to eq. \eqref{eq: the detailled v^-1 relationship for p}:
\begin{equation}
v_{EP}^{-1}=\phi_{p}^{''}(\mu_{EP})+\sum_{i}\left[\phi_{i}^{(3)}(\mu_{EP})\frac{m_{3}^{i}}{2v_{EP}}\right]+\phi_{p}^{(4)}(\mu_{EP})\frac{3v_{EP}^{2}}{3!v_{EP}}+\O\left(n^{-1}\right)\label{eq: detailled v_EP ^-1 relationship}
\end{equation}
We can see where the EP approximation produces errors. The $\phi_{p}^{''}$
term is well approximated: since $\left|\mu-\mu_{EP}\right|=\O\left(n^{-2}\right)$,
we have $\phi_{p}^{''}(\mu)=\phi_{p}^{''}(\mu_{EP})+\O\left(n^{-1}\right)$.
The term involving $m_{4}$ is also well approximated, and we can
see that the only term that fails is the $m_{3}$ term. The order
1 error is thus entirely coming from this term, which shows that EP
performance suffers more from the skewness of the target distribution
than from its kurtosis.

Finally, note that, with our result, we can get some intuitions about
the quality of the EP approximation using other metrics. For example,
if the most interesting metric is the KL divergence $KL\left(p,q\right)$,
the excess KL divergence from using the EP approximation $q$ instead
of the true minimizer $q_{KL}$ (which has the same mean $\mu$ and
variance $v$ as $p$) is given by:
\begin{eqnarray}
\Delta KL=\int p\log\frac{q_{KL}}{q} & = & \int p(x)\left(-\frac{\left(x-\mu\right)^{2}}{2v}+\frac{\left(x-\mu_{EP}\right)^{2}}{2v_{EP}}-\frac{1}{2}\log\left(\frac{v}{v_{EP}}\right)\right)\label{eq: excess KL divergence}\\
 & = & \frac{1}{2}\left[\frac{v}{v_{EP}}-1-\log\left(\frac{v}{v_{EP}}\right)\right]+\frac{\left(\mu-\mu_{EP}\right)^{2}}{2v_{EP}}\\
 & \approx & \frac{1}{4}\left(\frac{v-v_{EP}}{v_{EP}}\right)^{2}+\frac{\left(\mu-\mu_{EP}\right)^{2}}{2v_{EP}}
\end{eqnarray}

which we recognize as $KL\left(q_{KL},q\right)$. A similar formula
gives the excess KL divergence from using the CGA instead of $q_{KL}$.
For both methods, the variance term is of order $n^{-2}$ (though
it should be smaller for EP), but the mean term is of order $n^{-3}$
for EP while it is of order $n^{-1}$ for the CGA. Once again, EP
is found to be the better approximation.

Finally, note that our bounds are quite pessimistic: the true value
might be a much better fit than we have predicted here. 

A first cause is the bounding of the derivatives of $\log(p)$ (eqs.
\eqref{eq: strictly LC for target distribution},\eqref{eq: slowly changing log for target distributoin}):
while those bounds are correct, they might prove to be very pessimistic.
For example, if the contributions from the sites to the higher-derivatives
cancel each other out, a much lower bound than $nK_{d}$ might apply.
Similarly, there might be another lower bound on the curvature much
higher than $n\beta_{m}$.

Another cause is the bounding of the variance from the curvature.
While applying Brascamp-Lieb requires the distribution to have high
log-curvature everywhere, a distribution with high-curvature close
to the mode and low-curvature in the tails still has very low variance:
in such a case, the Brascamp-Lieb bound is very pessimistic. 

In order to improve on our bounds, we will thus need to use tighter
bounds on the log-derivatives of the hybrids and of the target distribution,
but we will also need an extension of the Brascamp-Lieb result that
can deal with those cases where a distribution is strongly log-concave
around its mode but, in the tails, the log-curvature is much lower.

\section{Conclusion}

EP has been used for now quite some time without any theoretical concrete
guarantees on its performance. In this work, we provide explicit performance
bounds and show that EP is superior to the CGA, in the sense of giving
provably better approximations of the mean and variance. There are
now theoretical arguments for substituting EP to the CGA in a number
of practical problems where the gain in precision is worth the increased
computational cost. This work tackled the first steps in proving that
EP offers an appropriate approximation. Continuing in its tracks will
most likely lead to more general and less pessimistic bounds, but
it remains an open question how to quantify the quality of the approximation
using other distance measures. For example, it would be highly useful
for machine learning if one could show bounds on prediction error
when using EP. We believe that our approach should extend to more
general performance measures and plan to investigate this further
in the future.

\bibliographystyle{unsrtnat}
\bibliography{ref}

\newpage{}

\appendix

\part*{Supplementary information of ``Bounding errors of Expectation-Propagation''}

\section{Improving on the Brascamp-Lieb bound\label{sec:Improving-on-the Brascamp-Lieb bound}}

In this section, we detail our mathematical results concerning the
extension of the Brascamp-Lieb bound.

We will note $LC(x)=\exp\left(-\phi(x)\right)$ a log-concave distribution.
We assume that $\phi$ is strongly convex, and slowly changing, ie:

\begin{eqnarray}
\forall x\ \phi^{''}(x) & \geq & \beta_{m}\label{eq:APPENDIX strongly log-concave}\\
\forall d\in[3,4,5,6]\ \left|\phi^{(d)}(x)\right| & \leq & K_{d}\label{eq:APPENDIX slowly changing log}
\end{eqnarray}

\subsection{The original Brascamp-Lieb theorem}

Let $\mu_{LC}=E_{LC}\left(x\right)$ be the expected value of $LC$.
The original Brascamp-Lieb result \citeyearpar{BrascampLieb:BestConstantsYoungsIneq}
concerns bounding fractional centered moments of $LC$ by the corresponding
fractional moments of a Gaussian of variance $\beta_{m}^{-1}$, centered
at $\mu_{LC}$. Noting $g(x)=\mathcal{N}\left(x|\mu_{LC},\beta_{m}^{-1}\right)$
that Gaussian, we have:
\begin{equation}
\forall\alpha\geq1\ E_{LC}\left(\left|x-\mu_{LC}\right|^{\alpha}\right)\leq E_{g}\left(\left|x-\mu_{LC}\right|^{\alpha}\right)\label{eq:full brascamp-lieb result}
\end{equation}
However, we are not interested in their full result, but only in a
restricted version of it which only concerns even moments. This version
simply reads:
\begin{eqnarray}
\forall k\in\mathbb{N}\ m_{2k}=E_{LC}\left(\left|x-\mu_{LC}\right|^{2k}\right) & \leq & \left(2k-1\right)m_{2k-2}\beta_{m}^{-1}\label{eq:restricted Brascamp-Lieb RECURRENCE}\\
m_{2k} & \leq & \left(2k-1\right)!!\beta_{m}^{-k}\label{eq:restricted Brascamp-Lieb}
\end{eqnarray}
where $\left(2k-1\right)!!$ is the double-factorial: the product
of all odd terms between $1$ and $2k-1$. Eq. \eqref{eq:restricted Brascamp-Lieb RECURRENCE}
might be a new result. Note that equality only occurs when $f(x)=1$
and $LC$ is Gaussian. Note also that the bounds on the higher derivatives
of $\phi$ are not needed for this result, but only for our extension.

We offer here a proof of eq. \eqref{eq:restricted Brascamp-Lieb RECURRENCE}
(from which eq. \eqref{eq:restricted Brascamp-Lieb} is a trivial
consequence), which is slightly different from Brascamp \& Lieb's
original proof. We believe this proof to be original, though it is
still quite similar to the original proof.
\begin{proof}
Let's decompose $LC(x)$ into two parts:
\begin{itemize}
\item $g(x)=\mathcal{N}\left(x|\mu_{LC},\beta_{m}^{-1}\right)$ the bounding
Gaussian with same mean as $LC$
\item $f(x)=\frac{LC(x)}{g(x)}$ the remainder
\end{itemize}
$f$ is easily shown to be log-concave, which means that it is unimodal.
We will note $x^{\star}$ the mode of $f$. $f$ is increasing on
$]-\infty,x^{\star}]$ and decreasing on $[x^{\star},\infty[$. We
thus know the sign of $f^{'}(x)$:
\begin{equation}
\mbox{sign}\left(f^{'}(x)\right)=\mbox{sign}\left(x^{\star}-x\right)\label{eq: sign of f-prime}
\end{equation}

Consider the integral: $\int_{-\infty}^{+\infty}g(x)f^{'}(x)dx$.
By integration by parts (or by Stein's lemma), we have:
\begin{eqnarray}
\int_{-\infty}^{+\infty}g(x)f^{'}(x)dx & = & \int_{-\infty}^{+\infty}g(x)f(x)\beta_{m}(x-\mu)dx\nonumber \\
 & = & \beta_{m}(\mu-\mu)\nonumber \\
\int_{-\infty}^{+\infty}g(x)f^{'}(x)dx & = & 0\label{eq: integral of g * f-prime =00003D 0}
\end{eqnarray}

We now split the integral at $\mu_{LC}$ and $x^{\star}$, assuming
without loss of generality that $x^{\star}\leq\mu_{LC}$:
\begin{eqnarray}
\int_{-\infty}^{x^{\star}}gf^{'}+\int_{x^{\star}}^{\mu_{LC}}gf^{'}+\int_{\mu_{LC}}^{\infty}gf^{'} & = & 0\nonumber \\
\int_{-\infty}^{x^{\star}}gf^{'}+\int_{x^{\star}}^{\mu_{LC}}gf^{'} & = & -\int_{\mu_{LC}}^{\infty}gf^{'}\nonumber \\
 & \geq & 0\label{eq: relationship between the different parts of int g f-prime}
\end{eqnarray}

Now consider a statistic $S_{k}(x)=\left(x-\mu_{LC}\right)^{2k-1}$.
Again using integration by parts, we have the following equality:
\begin{eqnarray}
\int g(x)f^{'}(x)S_{k}(x)dx & = & \int g(x)f(x)\left(\beta_{m}S_{k}(x)(x-\mu)-S_{k}^{'}(x)\right)dx\nonumber \\
\int g(x)f^{'}(x)\left(x-\mu_{LC}\right)^{2k-1}dx & = & \int LC(x)\left(\beta_{m}\left(x-\mu_{LC}\right)^{2k}-\left(2k-1\right)\left(x-\mu_{LC}\right)^{2k-2}\right)\nonumber \\
 & = & \beta_{m}m_{2k}-(2k-1)m_{2k-2}\label{eq: relationship between differents even moments}
\end{eqnarray}

At this point, we only need to prove that $\int gf^{'}S_{k}\leq0$
to finish our proof, from eq. \eqref{eq: relationship between differents even moments}.
We will actually prove a slightly stronger result: that even if we
cut the integral at $\mu_{LC}$, both halves are still negative:
\begin{eqnarray}
\int_{-\infty}^{\mu_{LC}}gf^{'}S_{k} & \leq & 0\label{eq: lower half (difficult)}\\
\int_{\mu_{LC}}^{\infty}gf^{'}S_{k} & \leq & 0\label{eq: upper half (easy)}
\end{eqnarray}

Eq. \eqref{eq: upper half (easy)} is trivial. $g$ is positive everywhere,
while $S_{k}(x)\geq0$ and $f^{'}(x)\leq0$ for $x\geq\mu_{LC}$.

Eq. \eqref{eq: lower half (difficult)} is slightly harder. From eq.
\eqref{eq: relationship between the different parts of int g f-prime},$\int_{-\infty}^{x^{\star}}gf^{'}+\int_{x^{\star}}^{\mu_{LC}}gf^{'}\geq0$,
where the first term is positive, and the second negative. When we
multiply the integrand by the decreasing positive function $-S_{k}(x)=-\left(x-\mu_{LC}\right)^{2k-1}$,
the order in the terms is preserved. To say it in equations:
\begin{eqnarray}
\int_{-\infty}^{x^{\star}}gf^{'}\left(-S_{k}\right) & \geq & \left(-\left(x^{\star}-\mu_{LC}\right)^{2k-1}\right)\int_{-\infty}^{x^{\star}}gf^{'}\nonumber \\
 & \geq & \left(-\left(x^{\star}-\mu_{LC}\right)^{2k-1}\right)\left(-\int_{x^{\star}}^{\mu_{LC}}gf^{'}\right)\nonumber \\
 & \geq & \int_{x^{\star}}^{\mu_{LC}}gf^{'}S_{k}
\end{eqnarray}

from which we finally find eq. \eqref{eq: lower half (difficult)},
which concludes our proof. Note that there is the equality $\int gf^{'}S_{k}=0$
IFF $f^{'}(x)=0$, justifying our earlier comment about $m_{2k}=(2k-1)\beta_{m}^{-1}m_{2k-2}$
IFF $LC(x)=g(x)$.
\end{proof}

\subsection{Extending the Brascamp-Lieb theorem}

The original Brascamp-Lieb result tells us that the spread of $LC(x)$
(as measured by its even moments) can't be too important, but it doesn't
tell us whether such distributions are close to being Gaussian, which
is what EP requires. By constraining the higher derivatives of $\phi(x)$,
we are able to constrain how far $LC$ is from a Gaussian distribution.
This is the essence of our extension of the Brascamp-Lieb theorem.
We derived the following:
\begin{thm}
Extension of the Brascamp-Lieb theorem

With $LC$ a strongly log-concave distribution with slowly changing
log-function (eqs. \eqref{eq:APPENDIX strongly log-concave}, \eqref{eq:APPENDIX slowly changing log}),
we have the following inequalities:

\begin{eqnarray}
\left|\phi^{'}\left(\mu_{LC}\right)\right| & \leq & \frac{K_{3}}{2\beta_{m}}\label{eq: appendix phi-prime bound}\\
\left|\frac{m_{3}}{m_{2}}\right| & \leq & 2\frac{K_{3}}{\beta_{m}^{2}}\label{eq: appendix m3 / m2 bound}\\
\left|\frac{m_{5}}{m_{2}}\right| & \leq & \frac{17K_{3}}{\beta_{m}^{3}}\label{eq: appendix m5/m2 bound}
\end{eqnarray}

which generalizes to:
\begin{equation}
\left|\frac{m_{2k+1}}{m_{2}}\right|\leq C_{k}\frac{K_{3}}{\beta_{m}^{k+1}}
\end{equation}

The following first order expansions of $m_{2}$, $m_{3}$ and $m_{4}$:
\begin{eqnarray}
\left|m_{2}^{-1}-\phi_{2}^{''}(\mu_{LC})\right| & \leq & \frac{K_{3}^{2}}{\beta_{m}^{2}}+\frac{K_{4}}{2\beta_{m}}\label{eq: appendix m_2 INVERSE relationship}\\
\left|\phi^{''}(\mu_{LC})m_{2}-1\right| & \leq & \frac{K_{3}^{2}}{\beta_{m}^{3}}+\frac{K_{4}}{2\beta_{m}^{2}}\label{eq: appendix m2 first order expansion}\\
\left|\phi^{''}(\mu_{LC})m_{3}+\left(\phi^{'}(\mu_{LC})m_{2}+\frac{\phi^{(3)}(\mu_{LC})}{2}m_{4}\right)\right| & \leq & \frac{17}{6}\frac{K_{3}K_{4}}{\beta_{m}^{4}}+\frac{5}{8}\frac{K_{5}}{\beta_{m}^{3}}\label{eq: appendix m3 first order expansion}\\
\left|\phi^{''}(\mu_{LC})m_{4}-3m_{2}\right| & \leq & \frac{19}{2}\frac{K_{3}^{2}}{\beta_{m}^{4}}+\frac{5}{2}\frac{K_{4}}{\beta_{m}^{3}}\label{eq: appendix m4 first order expansion}
\end{eqnarray}

which generalizes to:
\begin{eqnarray}
m_{2k+2} & \approx & \frac{\left(2k+1\right)}{\phi^{''}(\mu_{LC})}m_{2k}\\
 & \approx & \left(2k+1\right)!!\left[\phi^{''}(\mu_{LC})\right]^{-(k+1)}
\end{eqnarray}

And the following higher order relationships:
\begin{eqnarray}
\left|\phi^{'}(\mu_{LC})+\frac{\phi^{(3)}(\mu_{LC})}{2}m_{2}\right| & \leq & \frac{K_{3}K_{4}}{3\beta_{m}^{3}}+\frac{K_{5}}{8\beta_{m}^{2}}\label{eq: appendix second order e relationship}\\
\left|m_{2}^{-1}-\phi^{''}(\mu_{LC})-\frac{\phi^{(3)}(\mu_{LC})}{2}\frac{m_{3}}{m_{2}}-\frac{\phi^{(4)}(\mu_{LC})}{3!}\frac{m_{4}}{m_{2}}\right| & \leq & \frac{17}{24}\frac{K_{3}K_{5}}{\beta_{m}^{3}}+\frac{K_{6}}{8\beta_{m}^{2}}\label{eq: appendix second order v^-1 relationship}
\end{eqnarray}

\end{thm}
Note that we refer to eq. \eqref{eq: appendix m_2 INVERSE relationship},
\eqref{eq: appendix m3 first order expansion} and \eqref{eq: appendix m4 first order expansion}
as first order expansions because you can read them as, respectively:

\begin{eqnarray*}
m_{2} & \approx & \left(\phi^{''}(\mu_{LC})\right)^{-1}\\
m_{3} & \approx & -\left(\phi^{''}(\mu_{LC})\right)^{-1}\left(\phi^{'}(\mu_{LC})m_{2}+\frac{\phi^{(3)}(\mu_{LC})}{2}m_{4}\right)\\
m_{4} & \approx & 3\left(\phi^{''}(\mu_{LC})\right)^{-1}m_{2}
\end{eqnarray*}

These relationships are not exhaustive, and one could find many such
relationships for even higher orders. The list presented here only
concerns results which we will need for our bound on EP.
\begin{proof}
We will first give an outline of the proof, and then dive into all
the equations of the full proof.

The key component of the proof is Stein's lemma (ie: integration by
parts). For $LC=\exp\left(-\phi(x)\right)$, it reads: for any statistic
$S(x)$ with at-most-polynomial growth:
\begin{equation}
E_{LC}\left(\phi^{'}(x)S(x)-S^{'}(x)\right)=0\label{eq: Stein's lemma general form}
\end{equation}

which we will only use for statistics of the form $S_{k}(x)=\left(x-\mu_{LC}\right)^{k}$.
This gives us the following relationships:
\begin{eqnarray}
E_{LC}\left(\phi^{'}(x)\right) & = & 0\label{eq: first stein}\\
E_{LC}\left(\phi^{'}(x)(x-\mu_{LC})\right) & = & 1\label{eq:second stein}\\
E_{LC}\left(\phi^{'}(x)(x-\mu_{LC})^{2}\right) & = & 0\label{eq:third stein}\\
E_{LC}\left(\phi^{'}(x)(x-\mu_{LC})^{3}\right) & = & 3m_{2}\label{eq:fourth stein}
\end{eqnarray}

and further relationships of the same form that we won't need. The
key intuition in understanding why $LC$ is almost Gaussian is the
following: $\phi^{'}(x)\approx\phi^{''}(\mu_{LC})(x-\mu)$. The Stein
relationships for $LC$ are thus almost the same relationships that
would be obeyed by the Gaussian $g_{\mu_{LC}}(x)=\mathcal{N}\left(x|\mu_{LC},\left(\phi^{''}(\mu_{LC})\right)^{-1}\right)$.
This is why $LC$ is close to $g_{\mu_{LC}}$.

For all these relationships, we will perform a Taylor expansion around
$\mu_{LC}$, which now gives us self-consistency relationships between
the different moments of $LC$. For example, just keeping the first
term in eq. \eqref{eq: first stein} gives us eq. \eqref{eq: appendix phi-prime bound}:

\[
\phi^{'}\left(\mu_{LC}\right)\approx0
\]

We need to be careful with how we deal with the remainder of the Taylor
approximation. Using the Taylor-Lagrange formula, we can bound the
error that results from cutting off the Taylor series after some term,
with a term of the form $C\times\left(x-\mu\right)^{k}$ for some
constant C. The expected value under $LC$ of that term can then bounded
from the Brascamp-Lieb theorem. For example, to perform the cut-off
of eq. \eqref{eq: first stein} we just did, we start from the Taylor-Lagrange
expression:
\begin{equation}
\left|\phi^{'}(x)-\phi^{'}(\mu_{LC})-\phi^{''}(\mu_{LC})(x-\mu_{LC})\right|\leq\frac{K_{3}}{2}\left(x-\mu_{LC}\right)^{2}\label{eq: start of phi-prime bound}
\end{equation}

which, when we take the expected value, becomes:
\begin{equation}
\left|E_{LC}\left(\phi^{'}(x)\right)-\phi^{'}(\mu_{LC})\right|\leq\frac{K_{3}}{2}m_{2}\leq\frac{K_{3}}{2\beta_{m}}\label{eq: end of phi-prime bound}
\end{equation}

where we have applied the Brascamp-Lieb theorem. This concludes the
proof of eq. \eqref{eq: appendix phi-prime bound}, and our introduction
to the full proof.

Let's now prove the second relationship of the theorem: eq. \eqref{eq: appendix m3 / m2 bound}.
We start from eq. \eqref{eq:third stein}. We perform the expansion
of $\phi^{'}(x)$ up to the $\phi^{''}(\mu_{LC})(x-\mu_{LC})$ term.
From Taylor-Lagrange, the error is:
\begin{eqnarray}
\left|\phi^{'}(x)-\phi'(\mu_{LC})-\phi^{''}(\mu_{LC})(x-\mu_{LC})\right| & \leq & \frac{K_{3}}{2}\left(x-\mu_{LC}\right)^{2}\nonumber \\
\left|\phi^{'}(x)\left(x-\mu_{LC}\right)^{2}-\phi'(\mu_{LC})\left(x-\mu_{LC}\right)^{2}-\phi^{''}(\mu_{LC})\left(x-\mu_{LC}\right)^{3}\right| & \leq & \frac{K_{3}}{2}\left(x-\mu_{LC}\right)^{4}\label{eq: PIKACHU}
\end{eqnarray}

We now take the expected value:
\begin{equation}
\left|\phi'(\mu_{LC})m_{2}+\phi^{''}(\mu_{LC})m_{3}-E_{LC}\left(\phi^{'}(x)\right)\right|\leq\frac{K_{3}}{2}m_{4}
\end{equation}

Finally, we divide by $m_{2}$, take out the $\phi^{'}(\mu_{LC})$
term from the absolute value, use the bound on $\frac{m_{4}}{m_{2}}$
from eq. \eqref{eq:restricted Brascamp-Lieb RECURRENCE}, and lower
bound $\phi^{''}(\mu_{LC})$:
\begin{eqnarray}
\left|\phi^{''}(\mu_{LC})\frac{m_{3}}{m_{2}}\right| & \leq & \frac{K_{3}}{2}\frac{m_{4}}{m_{2}}+\left|\phi^{'}(\mu_{LC})\right|\nonumber \\
 & \leq & \frac{K_{3}}{2\beta_{m}}(3)+\frac{K_{3}}{2\beta_{m}}\nonumber \\
 & \leq & \frac{2K_{3}}{\beta_{m}}\\
\left|\frac{m_{3}}{m_{2}}\right| & \leq & \frac{2K_{3}}{\beta_{m}^{2}}\label{eq: bound on m3 / m2 APPENDIX}
\end{eqnarray}

which gives us eq. \eqref{eq: appendix m3 / m2 bound}.

Now, let's prove the bound on $m_{5}$ (eq. \eqref{eq: appendix m5/m2 bound}).
The demonstration is quite similar to the $m_{3}$ bound. We start
from another Stein relationship:
\[
E_{LC}\left(\phi^{'}(x)\left(x-\mu_{LC}\right)^{4}\right)=4m_{3}
\]

With the same Taylor-Lagrange expansion as in eq. \eqref{eq: PIKACHU}
and after taking the expected value, we have:
\begin{equation}
\left|4m_{3}-\phi^{'}(\mu_{LC})m_{4}-\phi^{''}(\mu_{LC})m_{5}\right|\leq\frac{K_{3}}{2}m_{6}
\end{equation}

Which we divide by $m_{2}$ and manipulate further:
\begin{eqnarray}
\left|\phi^{''}(\mu_{LC})\frac{m_{5}}{m_{2}}\right| & \leq & 4\left|\frac{m_{3}}{m_{2}}\right|+\left|\phi^{'}(\mu_{LC})\right|\frac{m_{4}}{m_{2}}+\frac{K_{3}}{2}\frac{m_{6}}{m_{2}}\nonumber \\
 & \leq & \frac{8K_{3}}{\beta_{m}^{2}}+\frac{K_{3}}{2\beta_{m}}\frac{3}{\beta_{m}}+\frac{K_{3}}{2}\frac{15}{\beta_{m}^{2}}\nonumber \\
 & \leq & \frac{17K_{3}}{\beta_{m}^{2}}\nonumber \\
\left|\frac{m_{5}}{m_{2}}\right| & \leq & \frac{17K_{3}}{\beta_{m}^{3}}
\end{eqnarray}

which gives us eq. \eqref{eq: appendix m5/m2 bound}.

In order to show that any odd centered moment admits a similar bound
(as we mention it the main text), we proceed by induction. The Stein
relationships:
\[
E\left(\phi^{'}(x)\left(x-\mu_{LC}\right)^{2k}-2k\left(x-\mu_{LC}\right)^{2k-1}\right)=0
\]
give us the inductive step through steps identical to the preceeding
equations, and we have already have the initialization (from eq. \ref{eq: bound on m3 / m2 APPENDIX}).
We can thus find similar bounds for any higher odd moment of $LC(x)$.

Now we will prove the first order expansions, starting with the one
for $m_{2}$ (eq. \eqref{eq: appendix m_2 INVERSE relationship}).
We now start from eq. \eqref{eq:second stein}, which is:
\[
E_{LC}\left(\phi^{'}(x)\left(x-\mu_{LC}\right)\right)=1
\]

First step, the Taylor-Lagrange expansion. We cut off the Taylor series
at $\frac{\phi^{(3)}(\mu_{LC})}{2}\left(x-\mu_{LC}\right)^{2}$. We
can bound the error with:
\begin{equation}
\left|\phi^{'}(x)\left(x-\mu_{LC}\right)-\phi^{'}(\mu_{LC})\left(x-\mu_{LC}\right)-\phi^{''}(\mu_{LC})\left(x-\mu_{LC}\right)^{2}-\frac{\phi^{(3)}(\mu_{LC})}{2}\left(x-\mu_{LC}\right)^{3}\right|\leq\frac{K_{4}}{3!}\left(x-\mu_{LC}\right)^{4}
\end{equation}

which becomes, when we take the expected value:
\begin{eqnarray}
\left|1-0-\phi^{''}(\mu_{LC})m_{2}-\frac{\phi^{(3)}(\mu_{LC})}{2}m_{3}\right| & \leq & \frac{K_{4}}{3!}m_{4}\nonumber \\
\left|\phi^{''}(\mu_{LC})m_{2}-1\right| & \leq & \frac{1}{2}\left|\phi^{(3)}(\mu_{LC})\right|\left|m_{3}\right|+\frac{K_{4}}{3!}m_{4}\nonumber \\
\left|m_{2}^{-1}-\phi^{''}(\mu_{LC})\right| & \leq & \frac{1}{2}\left|\phi^{(3)}(\mu_{LC})\right|\left|\frac{m_{3}}{m_{2}}\right|+\frac{K_{4}}{3!}\frac{m_{4}}{m_{2}}\\
 & \leq & \frac{K_{3}^{2}}{\beta_{m}^{2}}+\frac{K_{4}}{2\beta_{m}}
\end{eqnarray}

which proves eq. \eqref{eq: appendix m_2 INVERSE relationship}, from
which eq. \eqref{eq: appendix m2 first order expansion} is a trivial
consequence.

Now, the $m_{3}$ first order expansion (eq. \eqref{eq: appendix m3 first order expansion}).
We start from the Stein relationship from eq. \eqref{eq:third stein}
(which we already used to prove the bound on $\left|\frac{m_{3}}{m_{2}}\right|$).
\[
E_{LC}\left(\phi^{'}(x)\left(x-\mu_{LC}\right)^{2}\right)=0
\]

The difference between the $m_{3}$ bound and the $m_{3}$ first order
expansion is that we take a higher-order expansion of $\phi^{'}(x)$.
This time, we stop at $\phi^{(4)}(\mu_{LC})\left(x-\mu_{LC}\right)^{3}$.
The Taylor-Lagrange error is bounded by $\frac{K_{5}}{4!}\left(x-\mu_{LC}\right)^{4}$.
This gives us the following bound once we take the expected value.
\begin{equation}
\left|\phi^{'}(\mu_{LC})m_{2}+\phi^{''}(\mu_{LC})m_{3}+\frac{\phi^{(3)}(\mu_{LC})}{2}m_{4}+\frac{\phi^{(4)}(\mu_{LC})}{3!}m_{5}\right|\leq\frac{K_{5}}{4!}m_{6}
\end{equation}

In that equation, $m_{5}$ is an order of magnitude smaller than the
other terms, and we take it out of the absolute value:
\begin{eqnarray}
\left|\phi^{''}(\mu_{LC})m_{3}-\left(-\phi'(\mu_{LC})m_{2}-\frac{\phi^{(3)}(\mu_{LC})}{2}m_{4}\right)\right| & \leq & \frac{\left|\phi^{(4)}(\mu_{LC})\right|}{3!}\left|m_{5}\right|+\frac{K_{5}}{4!}m_{6}\nonumber \\
 & \leq & \frac{K_{4}}{3!}\frac{17K_{3}}{\beta_{m}^{4}}+\frac{K_{5}}{4!}\frac{15}{\beta_{m}^{3}}\nonumber \\
 & \leq & \frac{17}{6}\frac{K_{3}K_{4}}{\beta_{m}^{4}}+\frac{5}{8}\frac{K_{5}}{\beta_{m}^{3}}
\end{eqnarray}

which proves eq. \eqref{eq: appendix m3 first order expansion}.

Finally, we prove the last first order expansion: eq. \eqref{eq: appendix m4 first order expansion}
concerning $m_{4}$. We start from the last Stein relationship: eq.
\eqref{eq:fourth stein}:
\[
E_{LC}\left(\phi^{'}(x)\left(x-\mu_{LC}\right)^{3}\right)=3m_{2}
\]

We cut-off the Taylor series after $\frac{\phi^{(3)}(\mu_{LC})}{2}\left(x-\mu_{LC}\right)^{2}$.
After taking the expected value, the error is:
\begin{equation}
\left|3m_{2}-\phi^{'}(\mu_{LC})m_{3}-\phi^{''}(\mu_{LC})m_{4}-\frac{\phi^{(3)}(\mu_{LC})}{2}m_{5}\right|\leq\frac{K_{4}}{3!}m_{6}
\end{equation}

In this expression, $\phi^{'}(\mu_{LC})m_{3}$ and $\frac{\phi^{(3)}(\mu_{LC})}{2}m_{5}$
are both smaller by an order of magnitude, and we remove them from
the absolute value, to finally obtain:
\begin{eqnarray}
\left|\phi^{''}(\mu_{LC})m_{4}-3m_{2}\right| & \leq & \left|\phi^{'}(\mu_{LC})\right|\left|m_{3}\right|+\left|\frac{\phi^{(3)}(\mu_{LC})}{2}\right|\left|m_{5}\right|+\frac{K_{4}}{3!}m_{6}\nonumber \\
 & \leq & \frac{K_{3}}{2\beta_{m}}\frac{2K_{3}}{\beta_{m}^{3}}+\frac{K_{3}}{2}\frac{17K_{3}}{\beta_{m}^{4}}+\frac{K_{4}}{3!}\frac{15}{\beta_{m}^{3}}\nonumber \\
 & \leq & \frac{19}{2}\frac{K_{3}^{2}}{\beta_{m}^{4}}+\frac{5}{2}\frac{K_{4}}{\beta_{m}^{3}}
\end{eqnarray}

which proves eq. \eqref{eq: appendix m4 first order expansion}.

In order to find the first order developments of higher order even
moments, one proceeds identically to here but from the Stein relationships:
\begin{equation}
E\left(\phi^{'}(x)\left(x-\mu_{LC}\right)^{2k+1}-\left(2k+1\right)\left(x-\mu_{LC}\right)^{2k}\right)
\end{equation}
from which, by the same approach as the proof of eq. \ref{eq: appendix m4 first order expansion},
we have:
\begin{equation}
m_{2k+2}\approx\frac{\left(2k+1\right)}{\phi^{''}(\mu_{LC})}m_{2k}
\end{equation}
and by induction, we prove that:
\begin{equation}
m_{2k+2}\approx\left(2k+1\right)!!\left[\phi^{''}(\mu_{LC})\right]^{-(k+1)}
\end{equation}
which justifies our claim in the main text.

We are only left with proving the final two relationships. For eq.
\eqref{eq: appendix second order e relationship}, this corresponds
to doing a further expansion of the first Stein relationship (eq.
\eqref{eq: first stein}, from which we proved that $\phi^{'}(\mu_{LC})\approx0$):
\[
E_{LC}\left(\phi^{'}(x)\right)=0
\]

We stop the Taylor series after $\frac{\phi^{(4)}(\mu_{LC})}{3!}\left(x-\mu_{LC}\right)^{3}$.
After taking the expected value, we get:
\begin{equation}
\left|\phi^{'}(\mu_{LC})+\frac{\phi^{(3)}(\mu_{LC})}{2}m_{2}+\frac{\phi^{(4)}(\mu_{LC})}{3!}m_{3}\right|\leq\frac{K_{5}}{4!}m_{4}
\end{equation}

We extract the $m_{3}$ term which is an order of magnitude smaller
than the other ones, and obtain:
\begin{eqnarray}
\left|\phi^{'}(\mu_{LC})+\frac{\phi^{(3)}(\mu_{LC})}{2}m_{2}\right| & \leq & \left|\frac{\phi^{(4)}(\mu_{LC})}{3!}\right|\left|m_{3}\right|+\frac{K_{5}}{4!}m_{4}\nonumber \\
 & \leq & \frac{K_{4}}{3!}\frac{2K_{3}}{\beta_{m}^{3}}+\frac{K_{5}}{4!}\frac{3}{\beta_{m}^{2}}\nonumber \\
 & \leq & \frac{K_{3}K_{4}}{3\beta_{m}^{3}}+\frac{K_{5}}{8\beta_{m}^{2}}
\end{eqnarray}

which proves eq. \eqref{eq: appendix second order e relationship}.

At last, we reach the proof of eq. \eqref{eq: appendix second order v^-1 relationship}.
We start from the second Stein relationship (eq. \eqref{eq:second stein},
which we already used to get the first order expansion of $m_{2}$):
\[
E_{LC}\left(\phi^{'}(x)\left(x-\mu_{LC}\right)\right)=1
\]

We stop the Taylor series after $\frac{\phi^{(5)}(\mu_{LC})}{4!}\left(x-\mu_{LC}\right)^{4}$.
After taking the expected value, we get:
\begin{equation}
\left|1-\phi^{''}(\mu_{LC})m_{2}-\frac{\phi^{(3)}(\mu_{LC})}{2}m_{3}-\frac{\phi^{(4)}(\mu_{LC})}{3!}m_{4}-\frac{\phi^{(5)}(\mu_{LC})}{4!}m_{5}\right|\leq\frac{K_{6}}{5!}m_{6}
\end{equation}

We divide by $m_{2}$, then extract the $m_{5}$ term and obtain:
\begin{eqnarray}
\left|m_{2}^{-1}-\phi^{''}(\mu_{LC})+\frac{\phi^{(3)}(\mu_{LC})}{2}\frac{m_{3}}{m_{2}}+\frac{\phi^{(4)}(\mu_{LC})}{3!}\frac{m_{4}}{m_{2}}\right| & \leq & \left|\frac{\phi^{(5)}(\mu_{LC})}{4!}\right|\left|\frac{m_{5}}{m_{2}}\right|+\frac{K_{6}}{5!}\frac{m_{6}}{m_{2}}\nonumber \\
 & \leq & \frac{K_{5}}{4!}\frac{17K_{3}}{\beta_{m}^{3}}+\frac{K_{6}}{5!}\frac{15}{\beta_{m}^{2}}\nonumber \\
 & \leq & \frac{17}{24}\frac{K_{3}K_{5}}{\beta_{m}^{3}}+\frac{K_{6}}{8\beta_{m}^{2}}
\end{eqnarray}

proving eq. \eqref{eq: appendix second order v^-1 relationship} and
concluding our proof.
\end{proof}

\section{Quality of fixed-points of EP}

In this section, we give a detailed proof of our bounds on the quality
of the EP approximation.

We assume that all sites $f_{i}=\exp\left(-\phi_{i}(x)\right)$ are
$\beta_{m}$-strongly log-concave, with slowly changing log-functions.
That is:

\begin{eqnarray}
\forall i,x,\ \phi_{i}^{''}(x) & \geq & \beta_{m}\\
\forall d\in[3,4,5,6]\ \left|\phi_{d}^{(3)}(x)\right| & \leq & K_{d}
\end{eqnarray}

The target distribution $p(x)$ then inherits those properties from
the sites. Noting $\phi_{p}(x)=-\log\left(p(x)\right)=\sum_{i}\phi_{i}(x)$,
then $\phi_{p}$ is $n\beta_{m}$-strongly log-concave and for $d\in[3,4,5,6]$,
\begin{equation}
\left|\phi_{p}^{(d)}(x)\right|\leq nK_{d}
\end{equation}

Let $q_{i}\left(x|r_{i},\beta_{i}\right)$ be the site-approximations
of a fixed-point of EP, $q\left(x|r=\sum_{i}r_{i},\beta=\sum_{i}\beta_{i}\right)$
be the corresponding approximation of $p(x)$ and $h_{i}(x)$ the
corresponding hybrid distributions. From our hypothesis on the sites,
all hybrids are $\left(\beta_{m}+\beta_{-i}\right)$-strongly log-concave,
with slowly varying log-function (with constants $K_{d}$). We can
thus apply our results from section \ref{sec:Improving-on-the Brascamp-Lieb bound}
to all hybrids and the target distribution.

Some results to keep in mind on the hybrids: first of all, 
\begin{equation}
-\frac{\partial\log\left(h_{i}(x)\right)}{\partial x}=\phi_{i}^{'}(x)+\beta_{-i}x-r_{-i}\label{eq: phi_h prime}
\end{equation}
This expression is important as it is the one that appears in the
Stein relationships.

Also, because $q(x)$ is a Gaussian distribution of mean and variance
$\mu_{EP},v_{EP}$ and with natural parameters $r,\beta$:
\begin{eqnarray}
r & = & \beta\mu_{EP}\\
\beta & = & v_{EP}^{-1}
\end{eqnarray}

Finally, we have:
\begin{eqnarray}
\sum_{i}\beta_{-i}\mu_{EP} & = & \sum_{i,j\neq i}\beta_{j}\mu_{EP}\nonumber \\
 & = & \left(n-1\right)\sum_{j}\beta_{j}\mu_{EP}\nonumber \\
 & = & \left(n-1\right)\beta\mu_{EP}\nonumber \\
 & = & \left(n-1\right)r\nonumber \\
 & = & \left(n-1\right)\sum_{j}r_{j}\nonumber \\
 & = & \sum_{i,j\neq i}r_{j}\nonumber \\
\sum\beta_{-i}\mu_{EP} & = & \sum_{i}r_{-i}\label{eq: appendix sum of beta_-i vs sum of r_-i}
\end{eqnarray}

\subsection{Lower-bounding the $\beta_{i}$}

Let's show that we can lower bound the $\beta_{i}$ at the fixed-point
by $\beta_{m}$.

Recall that $\beta_{i}$ is obtained from the difference between the
inverse variance of $h_{i}(x)$ and $\beta_{-i}$, and $h_{i}(x)$
happens to be a $\left(\beta_{m}+\beta_{-i}\right)$-strongly log-concave
distribution. We can thus apply the Brascamp-Lieb inequality to the
variance:
\begin{eqnarray}
m_{2}^{i} & \leq & \frac{1}{\beta_{m}+\beta_{-i}}\\
\left(m_{2}^{i}\right)^{-1} & \geq & \beta_{m}+\beta_{-i}
\end{eqnarray}

Thus, $\beta_{i}=\left(m_{2}^{i}\right)^{-1}-\beta_{-i}\geq\beta_{m}$
and we have the claimed lower bound.%
\footnote{By the the same logic, if all sites are strongly log-concave, the
dynamics of EP must always maintain $\beta_{i}\geq\beta_{m}$. It
is thus useless to initialize the EP algorithm at a lower value.%
} 

Thus all hybrids are actually at least $n\beta_{m}$-strongly log-concave
(but could theoretically be stronger. This is one way our bounds can
be pessimistic).

\subsection{Approximation of various moments by $q(x)$ and the hybrids}

In this section, we will show that some moments of $p(x)$ are matched
approximately by the moments of $q(x)$ and/or the moments of the
hybrids $h_{i}(x)$.

We will note $m_{k}^{p}$ the $k^{th}$ centered moment of $p(x)$
and $m_{k}^{i}$ the moments of the hybrids. We will use $\mu,v$
for the mean and variance of $p(x)$ and $\mu_{EP},v_{EP}$ for the
mean and variance of $q(x)$ and all $h_{i}(x)$ (recall that, at
a fixed-point of EP, $q(x)$ and all $h_{i}(x)$ share the same mean
and variance). The mean and variance have gained special notation
due to their special status.

With these notations, the first three even moments of $q$ are respectively
$v_{EP}$, $3v_{EP}^{2}$ and $15v_{EP}^{3}$, while all odd moments
are $0$.

We will show that the following moments are matched:
\begin{thm}
When all sites are strongly log-concave with slowly changing log,
fixed-points of EP provide a good approximation of several moments
of $p(x)$:

\begin{eqnarray*}
\mu & = & \mu_{EP}+\O\left(n^{-2}\right)\\
v^{-1} & = & v_{EP}^{-1}+\O\left(1\right)\\
m_{3}^{p} & = & \sum_{i}m_{3}^{i}+\O\left(n^{-3}\right)\\
m_{4}^{p} & = & 3v_{EP}^{2}+\O\left(n^{-3}\right)\\
\forall i\ m_{4}^{p} & = & m_{4}^{i}+\O\left(n^{-3}\right)
\end{eqnarray*}
\end{thm}
\begin{proof}
Let's first give an outline of the proof.

The logic for all these results is similar. Because all hybrids $h_{i}(x)$
are $n\beta_{m}$-strongly log-concave with slowly changing-log, we
can apply the results of section \ref{sec:Improving-on-the Brascamp-Lieb bound}
on all those distributions, and obtain inequalities that relate the
moments of the $h_{i}(x)$ to one another. Since they all share the
same mean and variance, these become severely constrained. Since $p(x)$
is also log-concave with slowly changing log-function, its mean and
variance obey very similar relationships to $\mu_{EP}$ and $v_{EP}$.
From the fact that the pair $\left(\mu,v\right)$ and the pair $\left(\mu_{EP},v_{EP}\right)$
obey almost the same inequalities, we are able to deduce that they
are close to one another.

Let's start with $\mu$. From eq. \eqref{eq: appendix second order e relationship},
$\mu$ obeys the following simple relationship:
\begin{eqnarray}
\left|\phi_{p}^{'}(\mu)+\frac{\phi_{p}^{(3)}(\mu)}{2}v\right| & \leq & \frac{nK_{3}nK_{4}}{3n^{3}\beta_{m}^{3}}+\frac{nK_{5}}{8n^{2}\beta_{m}^{2}}\nonumber \\
 & \leq & n^{-1}\left(\frac{K_{3}K_{4}}{3\beta_{m}^{3}}+\frac{K_{5}}{8\beta_{m}^{2}}\right)\label{eq: second order relationship for e_p}
\end{eqnarray}

Applying the same results to all hybrids $h_{i}(x)$, we get:

\begin{eqnarray}
\forall i\ \left|\phi_{i}^{'}(\mu_{EP})+\beta_{-i}\mu_{EP}-r_{-i}+\frac{\phi_{i}^{(3)}(\mu_{EP})}{2}v_{EP}\right| & \leq & \frac{K_{3}K_{4}}{3n^{3}\beta_{m}^{3}}+\frac{K_{5}}{8n^{2}\beta_{m}^{2}}\nonumber \\
 & \leq & n^{-3}\frac{K_{3}K_{4}}{3\beta_{m}^{3}}+n^{-2}\frac{K_{5}}{8\beta_{m}^{2}}\label{eq: second order relationshiop for e_EP hybrid by hybrid}
\end{eqnarray}

which is slightly different than eq. \eqref{eq: second order relationship for e_p}.
Let's now sum the relationship obtained for each $h_{i}(x)$. The
$\beta_{-i}\mu_{EP}-r_{-i}$ terms drop out (eq. \eqref{eq: appendix sum of beta_-i vs sum of r_-i})
and we get:
\begin{equation}
\left|\phi_{p}^{'}(\mu_{EP})+\frac{\phi_{p}^{(3)}(\mu_{EP})}{2}v_{EP}\right|\leq n^{-2}\frac{K_{3}K_{4}}{3\beta_{m}^{3}}+n^{-1}\frac{K_{5}}{8\beta_{m}^{2}}\label{eq: final second order relationship for e_EP}
\end{equation}

We have that $\mu$ and $\mu_{EP}$ satisfy almost the same relationship
from eq. \eqref{eq: second order relationship for e_p} and \eqref{eq: final second order relationship for e_EP}.
We can use this to bound the distance between the two, as a function
of the distance between $v$ and $v_{EP}$:

\begin{eqnarray}
\phi_{p}^{'}(\mu_{EP})+\frac{\phi_{p}^{(3)}(\mu_{EP})}{2}v_{EP} & = & \phi_{p}^{'}(\mu_{EP})+\frac{\phi_{p}^{(3)}(\mu_{EP})}{2}(v+v_{EP}-v)\nonumber \\
\phi_{p}^{'}(\mu_{EP})+\frac{\phi_{p}^{(3)}(\mu_{EP})}{2}v_{EP}-\left(\phi_{p}^{'}(\mu)+\frac{\phi_{p}^{(3)}(\mu)}{2}v\right) & = & \left[\phi_{p}^{''}(\xi_{1})+\frac{\phi_{p}^{(4)}(\xi_{2})}{2}v\right](\mu_{EP}-\mu)\nonumber \\
 &  & \ \ +\frac{\phi_{p}^{(3)}(\mu_{EP})}{2}\left(v_{EP}-v\right)\label{eq: almost there !}
\end{eqnarray}

where $\xi_{1},\xi_{2}\in\left[\mu,\mu_{EP}\right]$ and we have used
first-order expansions at $\mu$ of $\phi_{p}^{'}(\mu_{EP})$ and
$\phi_{p}^{(3)}(\mu_{EP})$. We can go from upper bounding $\left[\phi_{p}^{''}(\xi_{1})+\frac{\phi_{p}^{(4)}(\xi_{2})}{2}v\right](\mu_{EP}-\mu)$
to upper bounding $\left|\mu-\mu_{EP}\right|$:
\begin{eqnarray}
\left|\left[\phi_{p}^{''}(\xi_{1})+\frac{\phi_{p}^{(4)}(\xi_{2})}{2}v\right](\mu_{EP}-\mu)\right| & \geq & \underset{\xi_{1},\xi_{2}}{\mbox{min}}\left(\left[\phi_{p}^{''}(\xi_{1})+\frac{\phi_{p}^{(4)}(\xi_{2})}{2}v\right]\right)\left|\mu-\mu_{EP}\right|\nonumber \\
 & \geq & \left[n\beta_{m}-\frac{K_{4}}{2\beta_{m}}\right]\left|\mu-\mu_{EP}\right|\label{eq:i need a name here}
\end{eqnarray}

We finally obtain a bound on the distance between $\mu$ and $\mu_{EP}$
by combining eqs. \eqref{eq: second order relationship for e_p},
\eqref{eq: final second order relationship for e_EP}, \eqref{eq: almost there !}
and \eqref{eq:i need a name here}:
\begin{eqnarray}
\left|\left[\phi_{p}^{''}(\xi_{1})+\frac{\phi_{p}^{(4)}(\xi_{2})}{2}v\right](\mu_{EP}-\mu)\right| & \leq & \left(n^{-1}+n^{-2}\right)\frac{K_{3}K_{4}}{3\beta_{m}^{3}}+2n^{-1}\frac{K_{5}}{8\beta_{m}^{2}}+n\frac{K_{3}}{2}\left|v-v_{EP}\right|\nonumber \\
 & \leq & \O\left(n^{-1}\right)+\O\left(n^{-1}\right)+\O\left(n\left|v-v_{EP}\right|\right)\\
\left|\mu-\mu_{EP}\right| & \leq & \O\left(n^{-2}\right)+\O\left(\left|v-v_{EP}\right|\right)\label{eq: final eq for e approx e_EP}
\end{eqnarray}

Once we show that $v=v_{EP}+\O\left(n^{-2}\right)$, eq. \eqref{eq: final eq for e approx e_EP}
will give us indeed that $\mu=\mu_{EP}+\O\left(n^{-2}\right)$.

Let's now show that $v\approx v_{EP}$. We start from the first order
expansion of $m_{2}^{-1}$ from our extension of the Brascamp-Lieb
theorem (eq. \eqref{eq: appendix m_2 INVERSE relationship}). For
$p(x)$, this gives us:
\begin{eqnarray}
\left|v^{-1}-\phi_{p}^{''}(\mu)\right| & \leq & \frac{n^{2}K_{3}^{2}}{n^{2}\beta_{m}^{2}}+\frac{nK_{4}}{2n\beta_{m}}\nonumber \\
 & \leq & \frac{K_{3}^{2}}{\beta_{m}^{2}}+\frac{K_{4}}{2\beta_{m}}\label{eq: first order of m_2 INV for p}
\end{eqnarray}

Again the corresponding relationship for the hybrids is not exactly
what we want it to be:
\begin{eqnarray}
\forall i\ \left|v_{EP}^{-1}-\phi_{i}^{''}(\mu_{EP})-\beta_{-i}\right| & \leq & \frac{K_{3}^{2}}{n^{2}\beta_{m}^{2}}+\frac{K_{4}}{2n\beta_{m}}\nonumber \\
 & \leq & n^{-2}\frac{K_{3}^{2}}{\beta_{m}^{2}}+n^{-1}\frac{K_{4}}{2\beta_{m}}\label{eq: first order of m_2 inv for each hybrid}
\end{eqnarray}

But again, we sum all those relationships:
\begin{equation}
\left|nv_{EP}^{-1}-\phi_{p}^{''}(\mu_{EP})-(n-1)\beta\right|\leq n^{-1}\frac{K_{3}^{2}}{\beta_{m}^{2}}+\frac{K_{4}}{2\beta_{m}}\label{eq: WABAGUNGA}
\end{equation}
which further simplifies, because $\beta=v_{EP}^{-1}$, into:
\begin{equation}
\left|v_{EP}^{-1}-\phi_{p}^{''}(\mu_{EP})\right|\leq n^{-1}\frac{K_{3}^{2}}{\beta_{m}^{2}}+\frac{K_{4}}{2\beta_{m}}\label{eq: first order of m_2 INV for EP}
\end{equation}

Again, we find that the pairs $\left(\mu,v\right)$ and $\left(\mu_{EP},v_{EP}\right)$
obey very similar relationships: eqs. \eqref{eq: first order of m_2 INV for p}
and \eqref{eq: first order of m_2 INV for EP}. We have:

\begin{equation}
\left|\phi_{p}^{''}(\mu)-\phi_{p}^{''}(\mu_{EP})\right|\leq K_{3}\left|\mu-\mu_{EP}\right|
\end{equation}

and this gives us that $v^{-1}\approx v_{EP}^{-1}$:
\begin{eqnarray}
\left|v^{-1}-v_{EP}^{-1}\right| & \leq & K_{3}\left|\mu-\mu_{EP}\right|+\left(1+n^{-1}\right)\frac{K_{3}^{2}}{\beta_{m}^{2}}+2\frac{K_{4}}{2\beta_{m}}\nonumber \\
\left|v^{-1}-v_{EP}^{-1}\right| & \leq & \O\left(1\right)+\O\left(n\left|\mu-\mu_{EP}\right|\right)\label{eq: final eq for v approx v_EP}
\end{eqnarray}
 Our final equations for the size of $\left|\mu-\mu_{EP}\right|$
and$\left|v^{-1}-v_{EP}^{-1}\right|$ seem to be caught in a loop:
you need to know how good one approximation is in order to know how
good the second will be and so on. This is not at all the case and
it is very easy to cut this loop. 

The easiest way is to remark that both $\mu$ and $\mu_{EP}$ are
$\O\left(n^{-1}\right)$ away from the mode of $p$ and so they must
be $\O\left(n^{-1}\right)$ from one another (see main text, section
\ref{sub:Computing-bounds-on EP }). This gives $v^{-1}=v_{EP}^{-1}+\O\left(1\right)$
(from eq. \eqref{eq: final eq for v approx v_EP}.

Then, we remark that both $v^{-1}$ and $v_{EP}^{-1}$ are order $n$.
The error for $\left|v-v_{EP}\right|$ is then of order $n^{-2}$
and we have that $\mu=\mu_{EP}+\O\left(n^{-2}\right)$, from eq. \eqref{eq: final second order relationship for e_EP}.
This concludes the first part of our proof.

Let's now look at the fourth moment of the target $m_{4}^{p}$. We
will show that is matched to by the fourth moment $m_{4}^{i}$ of
any hybrid and by the fourth moment of the Gaussian approximation
of $p(x)$: $3v_{EP}^{2}$.

From our Brascamp-Lieb extension, the first order approximation of
$m_{4}^{p}$ is:
\begin{equation}
\left|\phi^{''}(\mu)m_{4}^{p}-3v\right|\leq n^{-2}\left(\frac{19}{2}\frac{K_{3}^{2}}{\beta_{m}^{4}}+\frac{5}{2}\frac{K_{4}}{\beta_{m}^{3}}\right)\label{eq: first order for m_4_p}
\end{equation}

From which, intuitively: $m_{4}^{p}\approx3v\left(\phi_{p}^{''}(\mu)\right)^{-1}\approx3v^{2}\approx3v_{EP}^{2}$

Let's now formalize this intuition by bounding explicitely each error
term:
\begin{eqnarray}
3v^{2}-3v_{EP}^{2} & = & 6\left(v-v_{EP}\right)\frac{\left(v+v_{EP}\right)}{2}\\
\left|3v^{2}-3v_{EP}^{2}\right| & \leq & 6\left|v-v_{EP}\right|\frac{1}{2n\beta_{m}}\\
3v\left(\phi_{p}^{''}(\mu)\right)^{-1}-3v^{2} & = & 3v\left[\left(\phi_{p}^{''}(\mu)\right)^{-1}-v\right]\\
\left|3v\left(\phi_{p}^{''}(\mu)\right)^{-1}-3v^{2}\right| & \leq & \left|\left(\phi_{p}^{''}(\mu)\right)^{-1}-v\right|\frac{3}{n\beta_{m}}
\end{eqnarray}

Which we can bound using preceding relationships (eq. \eqref{eq: final eq for v approx v_EP}
and eq. \eqref{eq: appendix m2 first order expansion}), and which
gives us the final bound:
\begin{eqnarray}
\left|m_{4}^{p}-3v_{EP}^{2}\right| & \leq & n^{-3}\left(\frac{19}{2}\frac{K_{3}^{2}}{\beta_{m}^{5}}+\frac{5}{2}\frac{K_{4}}{\beta_{m}^{4}}\right)+\frac{6}{n\beta_{m}}\left|v-v_{EP}\right|+\frac{3}{n\beta_{m}}\frac{1}{n^{2}\beta_{m}^{2}}\left[\frac{2K_{3}^{2}}{\beta_{m}^{2}}+\frac{K_{4}}{2\beta_{m}}\right]\nonumber \\
 & \leq & \O\left(n^{-3}\right)
\end{eqnarray}

Let's note that this final approximation isn't any better of any worse,
in terms of orders of magnitude, than the original approximation $m_{4}^{p}\approx3v\left(\phi_{p}^{''}(\mu)\right)^{-1}$.

Another approximation that is of similar quality, in terms of orders
of magnitude, is for any hybrid $i$: $m_{4}^{p}\approx m_{4}^{i}$.
Indeed, from \ref{eq: appendix m4 first order expansion} (Brascamp-Lieb
extension: $m_{4}$ first order approximation), we have that:
\begin{equation}
\left|\left[\phi_{i}^{''}(\mu_{EP})+\beta_{-i}\right]m_{4}^{i}-3v_{EP}\right|\leq n^{-4}\frac{19}{2}\frac{K_{3}^{2}}{\beta_{m}^{4}}+n^{-3}\frac{5}{2}\frac{K_{4}}{\beta_{m}^{3}}\label{eq: first order for m_4 HYBRID}
\end{equation}

and see that $m_{4}^{i}$ would obey a similar relationship to $m_{4}^{p}$
(eq. \eqref{eq: first order for m_4_p}) if $\beta_{-i}\approx\sum_{j\neq i}\phi_{j}^{''}(\mu_{EP})$.
That happens to be the case because we also have:
\begin{equation}
\left|v_{EP}^{-1}-\left[\phi_{i}^{''}(\mu_{EP})+\beta_{-i}\right]\right|\leq n^{-2}\frac{2K_{3}^{2}}{\beta_{m}^{2}}+n^{-1}\frac{K_{4}}{2\beta_{m}}\label{eq: phi''_i + beta_-i is almost constant}
\end{equation}

Thus, $\phi_{i}^{''}(\mu_{EP})+\beta_{-i}$ is approximately constant
(in $i$), and approximately equal to $v_{EP}^{-1}$, which is an
important result in its own right. If we combine eqs. \eqref{eq: first order for m_4 HYBRID}
and \eqref{eq: phi''_i + beta_-i is almost constant}, we thus have:
\begin{eqnarray}
m_{4}^{i} & = & 3v_{EP}\left[\phi_{i}^{''}(\mu_{EP})+\beta_{-i}\right]^{-1}+\O\left(n^{-4}\right)\nonumber \\
 & = & 3v_{EP}^{2}+\O\left(n^{-4}\right)\nonumber \\
 & = & m_{4}^{p}+\O\left(n^{-3}\right)\label{eq: final M_4 hybrid approx m_4_P}
\end{eqnarray}

which concludes our proof that all fourth moments of the hybrids and
$q$ and $p$ are approximately equal. Note that an absolute error
of order $n^{-3}$ translates into a relative error of order $n^{-1}$.

Let's now show how to approximate the third moment of the target $m_{3}^{p}$
from the third moments of the hybrids $m_{3}^{i}$. We start for the
first-order approximation of $m_{3}^{p}$ (Brascamp-Lieb extension,
eq. \eqref{eq: appendix m3 first order expansion}):
\begin{eqnarray}
\left|\phi_{p}^{''}(\mu)m_{3}^{p}+\left(\phi^{'}(\mu)v+\frac{\phi^{(3)}(\mu)}{2}m_{4}^{p}\right)\right| & \leq & n^{-2}\left(\frac{17}{6}\frac{K_{3}K_{4}}{\beta_{m}^{4}}+\frac{5}{8}\frac{K_{5}}{\beta_{m}^{3}}\right)\label{eq: first order of m_3 for P}\\
m_{3}^{p} & \approx & -\left(\phi_{p}^{''}(\mu)\right)^{-1}\left(\phi^{'}(\mu)v+\frac{\phi^{(3)}(\mu)}{2}m_{4}^{p}\right)
\end{eqnarray}

For the hybrids, we have:
\begin{equation}
\forall i\ \left|\left(\phi_{i}^{''}(\mu_{EP})+\beta_{-i}\right)m_{3}^{i}+\left(\left(\phi_{i}^{'}(\mu_{EP})+\beta_{-i}\mu_{EP}-r_{-i}\right)v_{EP}+\frac{\phi_{i}^{(3)}(\mu_{EP})}{2}m_{4}^{i}\right)\right|\leq n^{-4}\frac{17}{6}\frac{K_{3}K_{4}}{\beta_{m}^{4}}+n^{-3}\frac{5}{8}\frac{K_{5}}{\beta_{m}^{3}}\label{eq: first order of m_3 for hybrids}
\end{equation}

We will perform the following steps:
\begin{eqnarray}
m_{3}^{i} & \approx & \left(\phi_{i}^{''}(\mu_{EP})+\beta_{-i}\right)^{-1}\left(\left(\phi_{i}^{'}(\mu_{EP})+\beta_{-i}\mu_{EP}-r_{-i}\right)v_{EP}+\frac{\phi_{i}^{(3)}(\mu_{EP})}{2}m_{4}^{i}\right)\\
 & \approx & v_{EP}\left(\left(\phi_{i}^{'}(\mu_{EP})+\beta_{-i}\mu_{EP}-r_{-i}\right)v_{EP}+\frac{\phi_{i}^{(3)}(\mu_{EP})}{2}m_{4}^{p}\right)
\end{eqnarray}

From which:
\begin{eqnarray}
\sum_{i}m_{3}^{i} & \approx & v_{EP}\sum\left(\left(\phi_{i}^{'}(\mu_{EP})+\beta_{-i}\mu_{EP}-r_{-i}\right)v_{EP}+\frac{\phi_{i}^{(3)}(\mu_{EP})}{2}m_{4}^{p}\right)\\
 & \approx & v_{EP}\left(\left(\phi_{p}^{'}(\mu_{EP})+0\right)v_{EP}+\frac{\phi_{p}^{(3)}(\mu_{EP})}{2}m_{4}^{p}\right)\label{eq: almost the same relatoinship !!!! bazinga}
\end{eqnarray}

from which we see that $m_{3}^{p}$ and $\sum_{i}m_{3}^{i}$ obey
very similar relationships (eq. \eqref{eq: first order of m_3 for P}
and eq. \eqref{eq: almost the same relatoinship !!!! bazinga}), and
can conclude that they are close.

More formally, starting from eq. \eqref{eq: first order of m_3 for hybrids},
let's replace $\phi_{i}^{''}(\mu_{EP})+\beta_{-i}$ with $\phi_{p}^{''}(\mu_{EP})$:
\begin{eqnarray}
m_{3}^{i} & = & \left(\phi_{i}^{''}(\mu_{EP})+\beta_{-i}\right)^{-1}\left(\left(\phi_{i}^{'}(\mu_{EP})+\beta_{-i}\mu_{EP}-r_{-i}\right)v_{EP}+\frac{\phi_{i}^{(3)}(\mu_{EP})}{2}m_{4}^{i}\right)+\O\left(n^{-4}\right)\nonumber \\
 & = & v_{EP}\left(\left(\phi_{i}^{'}(\mu_{EP})+\beta_{-i}\mu_{EP}-r_{-i}\right)v_{EP}+\frac{\phi_{i}^{(3)}(\mu_{EP})}{2}m_{4}^{i}\right)+\O\left(n^{-2}n^{-2}\right)+\O\left(n^{-4}\right)
\end{eqnarray}

Now, we replace $m_{4}^{i}$ with $m_{4}^{p}$. Since, $m_{4}^{i}=m_{4}^{p}+\O\left(n^{-3}\right)$,
we have:
\begin{equation}
m_{3}^{i}=v_{EP}\left(\left(\phi_{i}^{'}(\mu_{EP})+\beta_{-i}\mu_{EP}-r_{-i}\right)v_{EP}+\frac{\phi_{i}^{(3)}(\mu_{EP})}{2}m_{4}^{p}\right)+\O\left(n^{-4}\right)
\end{equation}

which we finally sum for $i$: the $\beta_{-i}\mu_{EP}-r_{-i}$ sum
to $0$, leaving:

\begin{equation}
\sum_{i}m_{3}^{i}=v_{EP}\left(\phi_{p}^{'}(\mu_{EP})v_{EP}+\frac{\phi_{p}^{(3)}(\mu_{EP})}{2}m_{4}^{p}\right)+\O\left(n^{-3}\right)\label{eq: almost there for m_3 equality}
\end{equation}

Because, $\mu=\mu_{EP}+\O\left(n^{-2}\right)$ and $v=\left(\phi_{p}^{''}(\mu)\right)^{-1}+\O\left(n^{-2}\right)=v_{EP}+\O\left(n^{-2}\right)$,
$\sum_{i}m_{3}^{i}$ and $m_{3}^{p}$ have identical first order expansions
(which is of order $n^{-2}$). More precisely:
\begin{eqnarray}
\phi_{p}^{'}(\mu_{EP})v_{EP}+\frac{\phi_{p}^{(3)}(\mu_{EP})}{2}m_{4}^{p} & = & \phi_{p}^{'}(\mu_{EP})v+\frac{\phi_{p}^{(3)}(\mu_{EP})}{2}m_{4}^{p}+\O\left(n^{-2}\right)\\
 & = & \phi^{'}(\mu)v+\frac{\phi^{(3)}(\mu)}{2}m_{4}^{p}+\O\left(n^{-2}\right)+\O\left(\left|\mu-\mu_{EP}\right|\right)
\end{eqnarray}
because: $\left|\phi^{'}(\mu)-\phi^{'}(\mu_{EP})\right|\leq n\beta_{m}\left|\mu-\mu_{EP}\right|$
and, similarly, $\phi^{(3)}(\mu)-\phi^{(3)}(\mu_{EP})=\O\left(n\left|\mu-\mu_{EP}\right|\right)$.
And:
\begin{eqnarray}
\left(v-v_{EP}\right)\left(\phi_{p}^{'}(\mu_{EP})v_{EP}+\frac{\phi_{p}^{(3)}(\mu_{EP})}{2}m_{4}^{p}\right) & = & \O\left(\left|v-v_{EP}\right|\right)\left(\O\left(1\right)\O\left(n^{-1}\right)+\O\left(n\right)\O\left(n^{-2}\right)\right)\nonumber \\
 & = & \O\left(n^{-3}\right)
\end{eqnarray}

Which gives us the final expression:

\begin{equation}
m_{3}=\sum_{i}m_{3}^{i}+\O\left(n^{-3}\right)\label{eq: approximation for m_3 =00003D sum m_3 ^i GOT THERE}
\end{equation}

which concludes our proofs on the quality of the EP approximation.

In the main text, we have also used the following relationship, detailing
the second order expansion of $v_{EP}^{-1}$:
\begin{equation}
v_{EP}^{-1}=\phi_{p}^{''}(\mu_{EP})+\sum_{i}\left[\phi_{i}^{(3)}(\mu_{EP})\frac{m_{3}^{i}}{2v_{EP}}\right]+\phi_{p}^{(4)}(\mu_{EP})\frac{3v_{EP}^{2}}{3!v_{EP}}+\O\left(n^{-1}\right)\label{eq: APPENDIX second order relationship for v_EP ^-1}
\end{equation}

For the inquisitive reader, this is obtained by starting from our
Brascamp-Lieb extension, eq. \eqref{eq: appendix second order v^-1 relationship},
applied to all hybrids. Then proceeding to approximate $m_{4}^{i}\approx3v_{EP}^{2}$
and summing. \end{proof}

\end{document}